\documentclass[twoside,12pt]{article}
\usepackage{amsmath,amssymb,amsthm,amsfonts,mathrsfs,wasysym,latexsym,times,lineno, subfigure,color,booktabs}
\usepackage{amsmath,amsfonts}
\usepackage{multirow}
\usepackage{array}
\usepackage{amsthm}
\usepackage{graphicx}
\usepackage{color}
\usepackage{multicol}
\usepackage{float}
\usepackage{cite}
\usepackage{url}
\usepackage{natbib}  
\usepackage{amsthm}

\newtheorem{thm}{Theorem}[section]

\newtheorem{lemma}[thm]{Lemma}

\newtheorem{theorem}[thm]{Theorem}

  \date{}

\title{Further Explanations on \\``SAT Requires Exhaustive Search''}
\author{Qingxiu Dong$^{1}$,\quad Guangyan Zhou$^{2}$,\quad Ke Xu$^{3}$$^{\dag}$\thanks{\dag Corresponding author}\\
\footnotesize $^{1}$School of Computer Science, Peking University \\
\footnotesize $^2$Department of Mathematics and Statistics, Beijing Technology and Business University\\
\footnotesize $^{3}$State Key Lab of Software Development Environment, Beihang University\\
\footnotesize kexu@buaa.edu.cn
}
\begin{document}
\maketitle

\begin{abstract}
Recently, ~\citet{sat} introduced a constructive approach for exploring computational hardness, proving that SAT requires exhaustive search.    
In light of certain misinterpretations concerning the contributions and proofs in that paper, we focus on providing detailed explanations in this work. 
We begin by delineating the core innovation of the constructive approach, shedding light on the pivotal concept of algorithm designability. 
We address the overlooked white-box diagonalization method and highlight the concept of an almost independent solution space.
% First, we highlight the primary innovation behind the constructive approach, including the foundational concept of algorithm designability, the overlooked white-box diagonalization method, as well as the idea of almost independent solution space.
In response to specific misunderstandings, such as the concerns surrounding the assumptions of Lemma 3.1, we offer comprehensive clarifications aimed at improving the comprehension of the proof.
We are grateful for the feedback received on our prior paper and hope this work can foster a more well-informed discussion.

\end{abstract}

\section{Background}
Currently, proving complexity lower bounds (e.g., a super-linear lower bound) for NP-complete problems has been challenging, leading to the failure of attempts to solve the P vs NP problem by demonstrating lower bounds for such problems~\citep{cook2000complexity}. The difficulty arises because the reduction-defined NP-complete problems, for example, the 3-SAT problem, have various effective strategies available for avoiding exhaustive search. %reduction-defined problems

Recently, ~\citet{sat} introduced a constructive approach for exploring computational hardness. This work demonstrates the fundamental difference between syntax and semantics by constructing extremely hard examples of constraint satisfaction problem (CSP) with large domains and SAT with long clauses and therefore proves that SAT requires exhaustive search.
Simultaneously, \citet{llm4science} conducted a pilot study on the P versus NP problem to explore the potential of large language models (LLMs) for scientific research.
In this paper, 
% Socratic reasoning, a general framework is proposed to encourage LLMs to recursively discover, solve, and integrate problems while facilitating self-evaluation and refinement.
% With this framework, 
LLM successfully produces a proof schema and engages in rigorous reasoning throughout 97 dialogue turns, in alignment with~\cite{sat}, shedding light on LLM for Science.

We welcome and value the attention and scrutiny from researchers across various fields on these two pieces of work. However, we have noticed some misunderstandings regarding certain arguments and proofs presented in our work. For instance, misunderstanding of the difference between the long clause SAT and the 3-SAT problems~\cite{evaluating}, the concerns regarding the assumption behind Lemma 3.1, etc.

In this paper, to address and resolve potential misunderstandings, as well as to foster a more informed discussion on the problem of P versus NP, we provide detailed explanations to tackle prevalent concerns that have been raised.
In particular, we highlight the primary innovation and conceptual cornerstones behind the proof, including the concept of algorithm designability, the white-box diagonalization methods, and the almost independent solution space.  
Furthermore, we clarify the key differences between SAT problems with extended clauses and those with 3-SAT clauses, and we address the issues that have been brought up regarding Lemma 3.1 and Lemma 3.2.
We hope this work can eliminate the remaining confusion arising from our earlier works and draw focus to the fundamental problems within the area of computational complexity.

% \section{Primary Innovation}
% \section{Conceptual cornerstone}
\section{Heart of Proof}
To better explain why \cite{sat} approaches the P vs NP problem in a constructive approach with Model RB, this section outlines the conceptual innovation and the main intuitions behind the proof.
First, we introduce the fundamental yet previously neglected concept -- algorithm designability.
Then, we distinguish the white-box diagonalization method and black-box diagonalization method for complexity lower bounds, laying the foundation for understanding the proof.
Last but not least, we highlight the idea of ``almost independent solution space'', which is central to understanding the computational hardness addressed by Model RB.
 
% We address the overlooked concept of algorithm designability, and the idea of ``almost independent solution space'', which is central to understanding the computational hardness addressed by Model RB.
% These concepts, together with the reconsideration of diagonalization methods for complexity lower bounds, lay the foundation for understanding the proof's novel strategy.

\subsection{New Concept: Algorithm Designability}
In theoretical computer science, previous endeavors traditionally focused on two core aspects: (1) computability, which queries the existence of an algorithm for a given problem; (2) computational complexity, which seeks to determine the feasibility of constructing an efficient algorithm. 

The proof of~\cite{sat} first focuses on a novel concept, designability. Designability involves a fundamental question: For a given problem, is it possible to conceive an algorithm? 
As mentioned in the second paragraph of the introduction of~\cite{sat}'s paper, the designability problem is similar to that proposed by David Hilbert in the early period of the 20th century if it is always possible to construct a finite formal system that can replace an area of mathematics containing infinitely many true mathematical statements.
When confronted with combinatorial challenges, exhaustive algorithms are inherently dictated by the definition of the problem itself, thereby negating the notion of `designed' algorithms. 
Consequently, the concept of algorithm designability is essentially replacing exhaustive algorithms with non-exhaustive ones.
With the concept of algorithm designability, we must prove the algorithm designability after proving the algorithm existence, rather than working on designing an efficient algorithm immediately. 

In short, algorithm designability is a more elemental and fundamental problem. However, it has been overlooked for decades, and the proof of ~\cite{sat}'s paper also explores this problem.

\subsection{Diagonalization Methods: White-Box or Black-Box?}
Currently, no partial result of complexity lower bounds (e.g. a super-linear lower bound) has been proved for any NP-Complete problem. This is very abnormal and signifies the possibility of a wrong path from the very beginning.

Both G\"{o}del and Turing have used the diagonalization method to prove impossibility results to answer Hilbert's three questions, which can be seen as the starting point of computer science. However, there is a key (but easily overlooked) difference between their proofs. Specifically, G\"{o}del's proof aims at a specific area of mathematics, where the (seemingly uneasy) constructive approach is indispensable. In contrast, Turing's proof lies in the abstract level of mathematics, with no need to construct a concrete program for reasoning. For convenience, these two proof strategies can be referred to as white-box and black-box diagonalization methods, respectively. Proving algorithmic impossibility results for a given problem is very similar to proving logical impossibility results for a specific area because they both involve reasoning about concrete examples. Thus, G\"{o}del's white-box diagonalization method is more suitable for proving lower bounds, which also naturally avoids various barriers (e.g. the relativization barrier and the natural proof barrier) originated and evolved from the (seemingly easy) choice of the black-box one in the beginning of complexity theory.

\subsection{Intuition behind the Proof: Almost Independent Solution Space}
%We aim to emphasize 
Another core insight behind the proof is the almost independent solution space. 
This intuition is crucial for understanding the construction of extremely hard problems based on the initial Model RB~\citep{xu2000exact}.

As explored in Section 3.1 of~\cite{llm4science}, when examining the P versus NP problem in computational theory, it is essential to offer intuitive reasons for why certain problems are computationally difficult.
The almost independence within the solution space offers a clear explanation for the computational hardness. 
It implies that there is no available strategy for effectively reducing the search space, which means we cannot skip any part of it in our search for a solution.
At the same time, the solution space's high degree of independence is a distinctive feature of Model RB, as highlighted in \cite{llm4science} and originally proposed by \citet{sat}. That’s the core intuition behind the proof.

% \section{Main Explanations}
\section{Explanation for Major Concerns}
In this section, we address the primary concerns associated with the study of \citet{sat} and \citet{llm4science}, providing detailed clarifications for each point raised.
We appreciate the interest and examination of our work by researchers from various disciplines. To clarify any confusion surrounding our arguments and proofs, we hope our explanations can facilitate better-informed discussion.

\subsection{Fundamental Differences between Long Clause SAT and 3-SAT}
Regarding known solutions of 3-SAT problems~\citep{evaluating}, we would like to clarify that our study focuses on the long clause SAT problem, which fundamentally differs from the 3-SAT problem.

\citet{evaluating}'s misunderstanding that ``their claim contradicts known results that it is possible to solve 3-SAT problems without exhaustive search'' seems to arise from a misinterpretation of our work. 
Indeed, as we specifically emphasized in the penultimate paragraph of~\citet{sat}, our findings do not apply to k-SAT problems with clauses of constant length, such as 3-SAT. 
Additionally, it has been mentioned in the abstract of~\cite{sat} that ``Specifically, proving lower bounds for many problems, such as 3-SAT, can be challenging because these problems have various effective strategies available for avoiding exhaustive search.'' 
Our results are not contradictory to the established understanding that 3-SAT problems can be solved without exhaustive search, albeit still requiring exponential time.

To avoid any potential misunderstandings, we wish to reiterate that the paper investigates instances of SAT with long clauses, and when we refer to the necessity of exhaustive search, we are speaking exclusively about solving these particular instances of SAT. Hence, the long clause SAT discussed in our paper is an entirely different scenario from that of 3-SAT.

\subsection{Independence of Lemma 3.1 from Model RB's Unique Properties}
Regarding the critique on subproblem properties~\citep{evaluating}, the misinterpretation mainly stems from a misinterpretation of the scope and application of Lemma 3.1 and Theorem 3.2 within~\cite{sat}. 

\citet{evaluating}'s main claim is that we assume, without proof, that each subproblem generated is also an instance of Model RB, which subsequently produces $d$ subproblems, thus necessitating an exhaustive search of $O(d^n)$. Furthermore, ~\citet{evaluating} argue that our subproblems lack the essential characteristic of Model RB.

\begin{lemma}\label{lm:subproblem}
If a CSP problem with $n$ variables and domain size $d$ can be solved in $T(n)=O(d^{cn})$ time ($0<c<1$ is a constant), then at most $O(d^c)$ subproblems with $n-1$ variables are needed to solve the original problem.
\end{lemma}

To address these concerns, we clarify that Lemma~\ref{lm:subproblem}, as shown above, is a discussion of the general approach to solving CSPs and does not rely on the unique properties of Model RB. ~\citet{evaluating}'s assertion that we assume each subproblem is an instance of Model RB is incorrect. Our analysis is far broader and encompasses all possible CSP instances with $n$ variables and domain size of $d$, not just those that adhere to Model RB.

In addition, the proof of Theorem 3.2 leverages only the initial level of decomposition, where the original problem is broken down into subproblems:

\begin{theorem}\label{th:main}
Model RB  cannot be solved in  $O(d^{c n})$ time for any constant $0<c<1$.
\end{theorem}
\begin{proof}
... Let $I$ be a RB instance with $n$ variables and $rn\ln d$ constraints. Suppose there exists some constant $0<c<1$ such that  $I$ can be solved in $O(d^{c n})$ time, then Lemma \ref{lm:subproblem} implies that $I$ can also be solved by assigning at most $O(d^c)$ values to an arbitrary variable ... This completes the proof of Theorem \ref{th:main}.
\end{proof}

It does not require, nor does it imply, a continual recursive decomposition into the subproblems of Model RB instances. \citet{evaluating} appears to have conflated the general methodology for CSP decomposition with the specific properties of Model RB.

The algorithm discussed in our paper is exact and deterministic. That means, for any instance with $n$ variables and a domain size of $d$, the algorithm deterministically decomposes the original problem into a certain number of subproblems, including but not limited to instances of Model RB. 
Lemma 3.1 aims to establish an upper bound on the number of subproblems for solving any CSP instance with these parameters.
Theorem 3.2 employs a proof by contradiction to demonstrate that the upper bound established is not consistent with the resolution of Model RB instances. It does not assert that each subproblem must exhibit unique properties of Model RB. Therefore, ~\citet{evaluating}'s concern regarding `proof on each of the subproblems is also an instance of Model RB' is unwarranted.

In a nutshell, our approach is designed to address the general case of CSPs, and the properties of Model RB are only utilized when illustrating the contradiction in Theorem 3.2. \citet{evaluating}'s critique is based on a misunderstanding of our method's scope and its application to Model RB.

\subsection{Dividing-and-conquering Algorithms for CSPs}
Regarding the main concerns on the assumption behind Lemma 3.1, we want to clarify that the only approach for exact solutions of CSPs and combinatorial problems is combinatorial search through the solution space.

CSPs are analogous to combinatorial problems, for which the current solution methods involve a combinatorial search through the solution space. This is an established and inescapable reality of solving such problems\footnote{\url{https://en.wikipedia.org/wiki/Combinatorial_search}}\footnote{\url{https://en.wikipedia.org/wiki/Constraint_satisfaction}}. Combinatorial searches represent the state-of-the-art in precisely solving general CSPs and combinatorial problems, and the current literature does not offer any alternative that escapes the necessity of exploring the search space to find exact solutions.

There are infinitely many algorithms that may emerge in the future, and it is impossible to exhaustively enumerate all of them. Regarding these potential future algorithms, it is inevitable to make assumptions that are consistent with the realities.
For example, without the Church-Turing thesis, the halting problem cannot be proved to be uncomputable in any machines. 

Therefore, we wish to clarify that the divide-and-conquer strategy in Lemma ~\ref{lm:subproblem} does not pertain to specific algorithms; rather, it denotes a combinatorial search through the solution space, which is the sole method for deriving exact solutions for CSPs.

\subsection{Assumptions are Necessary for Proving Lower Bounds}
In the discussions on the assumption presented in Section 3 ``Main Results'' of ~\cite{sat}'s paper, which involves dividing the original problem into subproblems, we aim to clarify that assumptions are necessary for reasoning and making conclusions about an infinite set of algorithms.

The form of algorithms is defined by Turing machines, while the essence of algorithms is how to use all possible mathematics to solve problems efficiently. Essentially, proving unconditional lower bounds for a given problem (including proving P $\neq$ NP) is using mathematics to question mathematics itself. This is a self-referential question independent of the axioms of mathematics. The fundamental reason is that self-referential statements about mathematical impossibility cannot be proved unconditionally, and even has no partial result. For example, consistency (i.e., impossibility of contradiction) is the most basic requirement for mathematics itself. However, as implied by G\"{o}del's second incompleteness theorem, mathematics, no matter how powerful it is, cannot prove its own consistency. This remarkable theorem, without any partial result, is established on underlying assumptions about the essence of formal proofs.

\section{Conclusion}
In this work, we aim to highlight the main contributions of our work and provide additional context to address concerns discussed in our prior research.
Therefore, we revisit the main concerns regarding \cite{sat} and \cite{llm4science}, providing further explanations of our intuitions and proof.

To better explain the theoretical motivation behind the proof in our papers, we re-interpreted the main innovation behind the proof from the perspectives of algorithm designability and the almost independent solution space of the Model RB, as well as G\"{o}del's white-box diagonalization method.
These ideas are important for comprehensively understanding why the constructive approach was chosen and how it works.
Regarding possible misunderstandings regarding that paper, we clarify the main arguments and proofs point by point.
Discussions on both of the papers are valued and appreciated, and we hope the explanation can foster a more informed discussion regarding the fundamental problem of P vs NP.

\bibliography{ref}
\bibliographystyle{plainnat}

\end{document}